\documentclass[runningheads]{llncs}
\usepackage{lineno}

\usepackage[T1]{fontenc}
\usepackage[utf8]{inputenc}
\usepackage{graphicx}
\usepackage{amsmath,amsfonts,amssymb}
\usepackage{thmtools} 					
\usepackage[sanserif,basic]{complexity}
\usepackage{xfrac}
\usepackage{xspace}
\usepackage{setspace}
\usepackage{lmodern}					

\usepackage{wrapfig}					

\usepackage{footnote}
\makesavenoteenv{tabular}
\makesavenoteenv{table}

\usepackage{adjustbox}

\usepackage{hyperref}
\usepackage[dvipsnames]{xcolor}

\usepackage
[
    noabbrev,   
    nameinlink, 
]
{cleveref} 

\newcommand*{\Otilde}{\widetilde{O}}

\newcommand*{\OHmega}{\mathrm{\Omega}}
\newcommand*{\nwspace}{\hspace*{.1em}}

\let\oldsqrt\sqrt
\def\hksqrt{\mathpalette\DHLhksqrt}
\def\DHLhksqrt#1#2{\setbox0=\hbox{$#1\oldsqrt{#2\,}$}\dimen0=\ht0
   \advance\dimen0-0.2\ht0
   \setbox2=\hbox{\vrule height\ht0 depth -\dimen0}
   {\box0\lower0.4pt\box2}}
\renewcommand\sqrt\hksqrt

\renewcommand{\leq}{\leqslant}
\renewcommand{\geq}{\geqslant}
\renewcommand{\le}{\leqslant}
\renewcommand{\ge}{\geqslant}
\newcommand{\eps}{\varepsilon}

\newcommand*{\G}{\mathcal{G}}
\renewcommand*{\P}{\mathcal{P}}

\begin{document}

\title{Compact Distance Oracles with Large Sensitivity and Low Stretch}

\author{Davide Bilò\inst{1} \and
	Keerti Choudhary\inst{2} \and
	Sarel Cohen\inst{3} \and
	Tobias Friedrich\inst{4} \and\\
	Simon Krogmann\inst{4} \and
	Martin Schirneck\inst{5}
}

\authorrunning{D.~Bilò et al.}

\institute{Department of Information Engineering, Computer Science and Mathematics, 
	University of L'Aquila, Italy \email{davide.bilo@univaq.it}\and
	Department of Computer Science and Engineering,\\
	Indian Institute of Technology Delhi, India \email{keerti@iitd.ac.in} \and
	School of Computer Science, The Academic College of Tel Aviv-Yaffo, 
	Israel \email{sarelco@mta.ac.il} \and
	Hasso Plattner Institute, University of Potsdam, Germany
	\email{firstname.lastname@hpi.de} \and
	Faculty of Computer Science, University of Vienna, Austria
	\email{martin.schirneck@univie.ac.at}
}

\maketitle              

\begin{abstract}
An \emph{$f$-edge fault-tolerant distance sensitive oracle} ($f$-DSO) with stretch $\sigma \geq 1$ 
is a data structure that preprocesses an input graph $G$.
When queried with the triple $(s,t,F)$, where $s, t \in V$ and $F \subseteq E$ contains at most $f$ edges of $G$, the oracle returns an estimate $\widehat{d}_{G-F}(s,t)$ of the distance $d_{G-F}(s,t)$ between $s$ and $t$ in the graph $G-F$ such that $d_{G-F}(s,t) \leq \widehat{d}_{G-F}(s,t) \leq \sigma \cdot d_{G-F}(s,t)$.

For any positive integer $k \ge 2$ and any $0 < \alpha < 1$, 
we present an $f$-DSO with sensitivity $f = o(\log n/\log\log n)$, 
stretch $2k-1$, space $O(n^{1+\frac{1}{k}+\alpha+o(1)})$, 
and an $\Otilde(n^{1+\frac{1}{k} - \frac{\alpha}{k(f+1)}})$
query time.

Prior to our work, there were only three known $f$-DSOs with subquadratic space.
The first one by Chechik et al.~[Algorithmica 2012]
has a stretch of $(8k-2)(f+1)$, depending on $f$. 
Another approach is storing an \emph{$f$-edge fault-tolerant $(2k{-}1)$-spanner} of $G$.
The bottleneck is the large query time due to the size of any such spanner,
which is $\OHmega(n^{1+1/k})$ under the Erd\H{o}s girth conjecture.
Bilò et al.~[STOC 2023] gave a solution with stretch $3+\varepsilon$,
query time $O(n^{\alpha})$ but space $O(n^{2-\frac{\alpha}{f+1}})$,
approaching the quadratic barrier for large sensitivity.

In the realm of subquadratic space, our $f$-DSOs are the first ones that guarantee, at the same time, large sensitivity, low stretch, and non-trivial query time.
To obtain our results, we use the approximate distance oracles of Thorup and Zwick [JACM 2005],
and the derandomization of the $f$-DSO of Weimann and Yuster [TALG 2013] 
that was recently given by Karthik and Parter [SODA 2021].

\end{abstract}

\setcounter{footnote}{0}

\section{Introduction}
\label{sec:info}

There are applications, like routing on edge devices,
where we want to quickly find out the distances
between pairs of vertices, but we cannot store the entire graph topology due to memory restrictions.
This problem is solved by a class of data structures called \emph{distance oracles} (DO).
Typically, not a single structure serves every use case and constructions
need to provide reasonable trade-offs between the space requirement, query time, and stretch,
that is, the quality of the estimated distance.
We are interested in the fault-tolerant setting.
Here, the data structure must additionally be able to tolerate multiple edge failures in the underlying graph. 
An \emph{$f$-edge fault-tolerant distance sensitivity oracles} \mbox{($f$-DSO)} for a graph $G = (V,E)$ is able to report,
for any two $s,t \in V$ and set $F \subseteq E$ of at most $f$ failing edges,
an estimate $\widehat{d}_{G-F}(s,t)$ of the \emph{replacement distance} $d_{G-F}(s,t)$
in the graph $G-F$. 
The parameter $f$ is the \emph{sensitivity} of the oracle.
We say the \emph{stretch} of the data structure is $\sigma$ if $d_{G-F}(s,t) \leq \widehat{d}_{G-F}(s,t) \leq \sigma \cdot d_{G-F}(s,t)$, for any admissible query $(s,t,F)$. 

Several $f$-DSOs with different space-stretch-time trade-offs have been designed in the last decades, most of which can only handle a very small number $f \le 2$ of failures \cite{BaswanaK13,BeKa08,BeKa09,BCFS21SingleSourceDSO_ESA,ChCo20,DemetrescuT02,DeThChRa08,DuanP09a,GrandoniVWilliamsFasterRPandDSO_journal,GuRen21,Ren22Improved}. We highlight those with sensitivity $f \geq 3$. 
The $f$-DSO of Duan and Ren \cite{DuRe22} requires $O(fn^4)$ space,\footnote{%
	The space is measured in the number of machine words on $O(\log n)$ bits.
}
returns exact distances in $f^{O(f)}$ query time, but the preprocessing algorithm that builds it 
requires exponential-in-$f$ time $n^{\OHmega(f)}$. 
The data structure by Chechik, Cohen, Fiat, and Kaplan~\cite{ChCoFiKa17} is more compact,
requiring $O(n^{2+o(1)}\log W)$ space, and can be preprocessed in time\footnote{%
	For a non-negative function $g(n)$,
	we use $\Otilde(g)$ to denote $O(g \cdot \textsf{polylog}(n))$.
} 
$\Otilde(n^{5+o(1)}\log W)$, where $W$ is the weight of the heaviest edge of $G$.
The oracle has stretch $1+\varepsilon$, for any constant $\varepsilon > 0$, with an $O(f^5\log n \log \log W)$ query time, and handles up to $f = o(\log n / \log \log n)$ failures. Finally, the $f$-DSO of Chechik, Langberg, Peleg, and Roditty~\cite{CLPR12} requires a subquadratic space of $O(fkn^{1+1/k}\log(nW))$, where $k\geq 1$ is an integer parameter, and has a fast query time of $\Otilde(f \log\log d_{G-F}(s,t))$, but guarantees only a stretch of $(8k-2)(f+1)$ that depends on the number $f$ of failures.

Another common way to provide approximate replacement distances 
in the presence of transient edge failures are fault-tolerant spanners~\cite{Levcopoulos98FaultTolerantGeometricSpanners}.
An \emph{$f$-edge fault-tolerant spanner with stretch $\sigma$ 
(fault-tolerant $\sigma$-spanner)} 
is a subgraph $H$ of $G$ such that
$d_{H-F}(s,t) \le \sigma \cdot d_{G-F}(s,t)$, for every suitable triple $(s,t,F)$, 
with $|F| \le f$.
For any positive integer $k$, 
Chechik, Langberg, Peleg, and Roditty~\cite{Chechik10FaultTolerantSpannersGeneralGraphs}
gave an algorithm computing a fault-tolerant $(2k{-}1)$-spanner with $O(f n^{1+1/k})$ edges.
This was recently improved by Bodwin, Dinitz, and Robelle 
by reducing the size to  $O(f^{1-1/k} n^{1+1/k})$~\cite{BodwinDinitzRobelle21OptimalVFTSpanners}
and eventually to $f^{1/2}n^{1+1/k} \cdot \poly(k)$ for any even $k$
and $f^{1/2-1/(2k)}n^{1+1/k} \cdot \poly(k)$ for odd 
$k$~\cite{BodwinDinitzRobelle22PartiallyOptimalEFTSpanners}.
The authors of the last work also show almost matching lower bounds of $\OHmega(f^{1/2-1/(2k)} n^{1+1/k} + fn)$ for general $k > 2$, 
and $\OHmega(f^{1/2}n^{3/2})$ for $k = 2$
assuming the Erdős girth conjecture~\cite{Erd64extremal}.

The main problem with the spanner approach is the high query time. 
In fact, to retrieve the approximate distance between a given pair of vertices,
one has to compute the single-source distance from one of them, say with Dijkstra's algorithm,
in time that is at least linear in the size of the spanner.
Therefore, an important problem in the field of fault-tolerant data structures
to design $f$-DSOs with subquadratic space, 
that simultaneously guarantee a non-trivial $o(n^{1+1/k})$ query time, a low stretch of $2k-1$,
and a large sensitivity $f$.\footnote{%
	Subquadratic space $f$-DSOs with stretch $2k-1$ can only exist for $k\geq 2$. 
	There is an $\OHmega(n^2)$-bit lower bound for exact $f$-DSOs,
	regardless of the query time~\cite{ThorupZ05}.
}

Very recently, Bilò, Chechik, Choudhary, Cohen, Friedrich, Krogmann, and Schirneck~\cite{Bilo23SubquadraticSpaceSTOC} addressed the same problem.
They presented, for all \mbox{$\varepsilon > 0$}, and constants $1/2 > \alpha > 0$, and $f$,
a $(3{+}\varepsilon)$-approximate $f$-DSO for unweighted graphs
taking space $\Otilde_{\varepsilon}(n^{2-\frac{\alpha}{f+1}} (\log n)^{f+1})$
and  has a query time of $\Otilde_{\varepsilon}(n^{\alpha})$.
While their query time is sub-linear, their space converges to quadratic for large sensitivity.

In contrast, we design a deterministic oracle for weighted graphs
that can handle up to $f = o(\log n/\log\log n)$ edge failures and 
provides a trade-off between stretch, space, and query time.
Namely, for any positive integer $k \ge 2$ and constant $1-\frac{1}{k} > \alpha > 0$,
our data structure has stretch $2k-1$,
requires $kn^{1+\alpha + \frac{1}{k} +o(1)}$ space, and can be queried in time $\Otilde(n^{1 + \frac{1}{k} -\frac{\alpha}{k(f+1)}})$. 
The query time improves to $\poly(D, f, \log n)$ for graphs in which the pair-wise hop distances are bounded by $D$.
If, for example, $D$ is polylogarithmic, the query time is as well.
We note that the query time of our $f$-DSO for general graphs is $\OHmega(n)$
for all choices of $\alpha$.

Both \cite{Bilo23SubquadraticSpaceSTOC} and this work approach
the problem by handling hop-short and hop-long paths separately,
as is common in the area~\cite{GrandoniVWilliamsFasterRPandDSO_journal,WY13},
and use the distance oracle of Thorup and Zwick~\cite{ThorupZ05} on the lowest level.
Apart from that, the techniques are different.
We highlight ours below.

\vspace*{.5em}
\noindent
\textbf{Our techniques.}
Our $f$-DSO for bounded hop diameter
is the result of combining the approximate distance oracles of Thorup and Zwick~\cite{ThorupZ05} 
with randomized replacement path covering (RPC), a collection of certain subgraphs of $G$,
introduced by Weimann and Yuster~\cite{WY13}.
Such coverings are very large, even larger than the underlying graph itself.
They are thus unusable when emphasizing subquadratic space,
barring additional processing.
The main issue when compressing an RPC is retaining the information
which subgraph is relevant for which query.
We provide two different ways to solve this.
One is based on the idea of using sparse spanners as proxies for the subgraphs in the covering, 
and the other one uses the recent derandomization technique of Karthik and Parter~\cite{KarthikParter21DeterministicRPC}.
To lift this to an arbitrary hop diameter,
we borrow from fault-tolerant spanners.
There, a single graph is constructed up front to handle all queries.
To achieve a compact oracle with a query time better than any spanner,
we instead turn this process around and 
use the hop-short $f$-DSO to combine only the subgraphs we need.

\vspace*{.5em}
\noindent
\textbf{Other related work.}
Demetrescu and Thorup \cite{DemetrescuT02} designed the first exact $1$-DSO for directed edge-weighted graphs with $O(n^2 \log n)$ space and $O(\log n)$ query time. Demetrescu, Thorup, Chowdhury, and Ramachandran~\cite{DeThChRa08} improved the query time to $O(1)$ and generalized the oracle to handle also a single vertex failure. Later, in two consecutive papers, Bernstein and Karger improved the preprocessing time from $\Otilde(mn^2)$ to $\Otilde(mn)$~\cite{BeKa08,BeKa09}.  Khanna and Baswana~\cite{KB10} designed $1$-DSO for unweighted graphs having size $O(k^5 n^{1+1/k} \frac{\log^3 n}{\eps^4})$, a stretch of $(2k-1)(1+\eps)$, and $O(k)$ query time. The problem of $1$-DSO was also studied with a special focus on the preprocessing time \cite{ChCo20,GrandoniVWilliamsFasterRPandDSO_journal,GuRen21,Ren22Improved,WY13}.

For the case of multiple failures, other than the results we explicitly mentioned in the introduction~\cite{ChCoFiKa17,CLPR12,DuRe22}, it is worth mentioning the $2$-DSO of Duan and Pettie~\cite{DuanP09a} with $O(n^2 \log^3 n)$ size and $O(\log n)$ query time and the work by van~den~Brand and Saranurak~\cite{BrSa19}.

\vspace*{.5em}
\noindent
\textbf{Outline.}
\Cref{sec:overview} provides an overview of our approach and presents our results.
The preliminaries and notation needed to follow the technical part are given in \Cref{sec:prelims}.
In \Cref{sec:small_diameter}, we first describe the randomized subquadratic-space $f$-DSO for short hop distances and then derandomize it, not only to obtain a deterministic construction but also to accelerate the query time to $\poly(D, f, \log n)$. 
\Cref{sec:large_diameter} then describes how to use this
to develop a deterministic subquadratic-space $f$-DSO also for hop-long replacement paths.

Some of the proofs are deferred to \Cref{app:omitted_proofs} due to space reasons.

\section{Overview}
\label{sec:overview}

Our first goal is to develop an $f$-DSO whose space is subquadratic in $n$, 
provided that the hop diameter\footnote{%
	The \emph{hop diameter} of a weighted graph is the minimum integer $D$ 
	such that all shortest paths between pairs of vertices have at most $D$ edges.
} 
$D$ and the sensitivity $f$ are not too large.
One of the first DSOs was given by Weimann and Yuster~\cite{WY13}.
It reports exact distances but, on graphs with a large hop diameter,
it is too large and too slow.
We first give an overview of their techniques and
then describe the steps we take to reduce the space as well as the query time using approximation.

Given the graph $G = (V,E)$ as well as positive integers $L$ and $f$,
the DSO in~\cite{WY13} samples a family $\G$ of $\Otilde(fL^f)$ random spanning subgraphs of $G$,
that is, all the subgraphs have the same vertex set $V$.
Each graph $G_i \in \G$ is generated by removing each edge of $G$ with probability $1/L$.
With high probability,\footnote{%
	An event occurs \emph{with high probability} (w.h.p.) if it has success probability 
	at least $1- n^{-c}$ for some constant $c >0$.
	In fact, $c$ can always be made arbitrarily large 
	without affecting the asymptotics.
} for all vertices $s,t \in V$ and sets $F \subseteq E$ of at most $f$ edge failures,
if there is a replacement path from $s$ to $t$ that has at most $L$ edges and none of them is in $F$, 
then there exists a subgraph $G_i \in \G$ that does not contain any edge of $F$ but such an replacement path.
Let $\G_F \subseteq \G$ be the subfamily of all the $G_i$ in which at least all of $F$ was removed.
In other words, if $s$ and $t$ have a \emph{hop-short} shortest path in $G-F$,
at least one of their replacement paths survives in a graph in $\G_F$.

To handle hop-short replacement paths, it is enough to go over the subgraphs
and report the minimum distance $d_{G_i}(s,t)$ over all $G_i \in \G_F$.
For the \emph{hop-long} replacement paths on more than $L$ edges,
a random subset $B \subseteq V$ of $\Otilde(fn/L)$ of \emph{pivots} is sampled.
This way any hop-long replacement path decomposes into short subpaths
such that both endpoints are in $B$.
To answer a hop-long query $(s,t,F)$, a dense weighted graph $H^F$ is created on the vertex set
$V(H^F) = B \cup \{s,t\}$ such that for any two $u,v \in V(H^F)$ the edge $\{u,v\}$
has weight $\min_{G_i \in \G_F} d_{G_i}(u,v)$.
Those edges represent the subpaths.
The oracle's eventual answer to the query is the distance $d_{H^F}(s,t)$ in $H^F$.

The replacement distances reported by the DSO are exact w.h.p.
However, this approach has several drawbacks.
The most important one for us is that each of the graphs $G_i$ has $\OHmega(m)$ edges,
raising the space
to store them all to $\OHmega(fL^f m)$, which is super-quadratic in $n$ for dense graphs $G$.
Also, the query time is rather high,
the bottleneck is computing the weight of the $O(|B|^2) = \Otilde(f^2 n^2/L^2)$ edges of $H^F$ for the hop-long paths.

The key observation for improving this result in graphs with a small hop diameter is
that there \emph{all} replacement paths are hop-short.
Afek, Bremler-Barr, Kaplan, Cohen, and Merritt~\cite{Afek02RestorationbyPathConcatenation_journal} showed that 
for undirected, weights graphs $G$ and failure sets $F \subseteq E$ with $|F| \le f$,
every shortest path in $G-F$ is a concatenation of at most $f+1$ shortest paths in $G$
interleaved with at most $f$ edges.
So if $D$ is a bound on the hop diameter of $G$, the hop diameter of $G-F$ is at most $L = (f{+}1)D+f$. 
With this definition of $L$, we can safely ignore hop-long replacement paths.
Note that the assumption of $G$ being undirected is essential here:
The Afek et al.~\cite{Afek02RestorationbyPathConcatenation_journal} result fails in directed graphs.
Moreover, there is no hope for a subquadratic DSO in that case.
Thorup and Zwick~\cite{ThorupZ05} showed that every data structure
reporting pairwise distances in a directed graph must take $\OHmega(n^2)$ space.
This holds even if the data structure does not support a single edge failure and only provides an arbitrary finite approximation of the distance.

Nevertheless, we can use approximation in order to reduce the space of the DSO for undirected graphs.
Instead of storing the subgraphs $G_i$, we replace them by the \emph{distance oracle} (DO)
of Thorup and Zwick~\cite{ThorupZ05}.
For any positive integer $k$ and $G_i$, we get a DO of size $O(kn^{1+1/k})$ that,
when queried with two vertices $s,t$, reports the distance $d_{G_i}(s,t)$
but with a \emph{stretch} of $2k-1$.
That means, the returned value $\widehat{d}(s,t)$ satisfies
$d_{G_i}(s,t) \le \widehat{d}(s,t) \le (2k{-}1) \cdot d_{G_i}(s,t)$.
While the use of more efficient data structures
reduces the space of our DSO to $\Otilde(fL^f n^{1+1/k})$,
discarding the actual subgraphs $G_i$ makes it impossible to recover
the information which edges have been removed in which graph, that is,
to compute the subfamily $\G_F$.
We provide two different ways to solve this.
The first one is to use spanners.
The DO in~\cite{ThorupZ05} is accompanied by a spanner of the same size
and we show that if the spanner associated with $G_i$ does not contain
an edge of $F$ then it is safe to rely on $G_i$ for the replacement distances,
even if the graph itself has some failing edges from $F$.

Interestingly, the other solution comes from derandomization.
Karthik and Parter~\cite{KarthikParter21DeterministicRPC} showed 
how to make the subgraph creation deterministic,
albeit now with $O((cfL \log m)^{f+1})$ such graphs for some constant $c > 0$.
This makes the resulting DSO less compact and also increases the preprocessing time.
However, they presented a way to compute the now deterministic family $\G_F$
using error-correcting codes.
This allows us to significantly improve the query time if the diameter is small.
For this, we show how to implement the encoding procedure without using additional storage space.

We present our results in the following setting.
We consider graphs with \emph{polynomial edge weights},
meaning that they are edge-weighted by positive reals from a range of size $\poly(n)$,
where $n$ is the number of vertices.
While the weights themselves may have arbitrary precision,
the number of values that can be written as sums of at most $n$ weights is again
polynomial.
Therefore, we can encode any graph distance in a constant number of $O(\log n)$-bit machine words.
The restriction on the range is justified as follows.
Let $W = \max_{e \in E} w(e)/\min_{e \in E} w(e)$ be the ratio between the maximum and minimum weight.
Chechik et al.~\cite[Lemma~4.1]{ChCoFiKa17} gave a reduction from approximate DSOs for general weighted graphs
to approximate DSOs for graphs with polynomial weights
that increases the space and preprocessing time only by a factor $O(\frac{\log W}{\log n})$,
the query time by a factor $O(\log\log W)$, and the stretch by a factor $1+\frac{1}{n}$.

In the statements below, $k$ controls the stretch vs.\ space trade-off
is an arbitrary positive integer, possibly even depending on the number of vertices $n$.
However, there are only space improvements to be had for values $k = O(\log n)$.

\begin{restatable}{theorem}{smalldiameter}
\label{thm:small_diameter}
	Let $G = (V,E)$ be an undirected graph with 
	polynomial edge weights, and hop diameter $D$.
	For all positive integers $k$ and $f = o(\log n/\log\log n)$,
	there is an $f$-DSO for $G$ that 
	has stretch $2k-1$ and satisfies the following properties.
	\begin{enumerate}
		\item (Randomized.) The DSO takes space $D^f kn^{1+\frac{1}{k} +o(1)}$,
			has a preprocessing time of $D^f kmn^{\frac{1}{k} +o(1)}$,
			and answers queries correctly w.h.p.\ in time $D^f n^{o(1)}$.\vspace*{.25em}
		\item (Deterministic.) The DSO takes space $D^{f+1} kn^{1+\frac{1}{k} +o(1)}$,
			has preprocessing time of $D^{f+1} kmn^{\frac{1}{k} +o(1)}$,
			and query time $O(f^3 D \nwspace \frac{\log n \log\log n}{\log D})$.
	\end{enumerate}
\end{restatable}

\vspace*{-.5em}

\begin{restatable}{corollary}{polylogdiameter}
\label{cor:polylog_diameter}
	If $G$ has a polylogarithmic hop diameter,
	then there is an $f$-DSO for $G$ with stretch $2k-1$ 
	that takes $kn^{1+\frac{1}{k} +o(1)}$ space, has a preprocessing time of $kmn^{\frac{1}{k} +o(1)}$, and $\Otilde(1)$ query time.
\end{restatable}

We also devise a solution for graphs with an arbitrarily large hop diameter.
To do so, we have to compute the correct distances for hop-long replacement paths.
In~\cite{WY13}, this was the role of the dense subgraph $H^F$ on the pivots in $B$.
Imagine we would sparsify it using the spanner construction above.
This would significantly reduce the number of edges we need
and stretch the distance $d_{H^F}(s,t)$ to at most $2k-1$ times the correct replacement distance.
But computing first the graph and then the distance would still take a lot of time.
Instead, the idea of our solution is
to prepare a spanner on vertex set $B$ for each subgraph and to combine only those we need for the result.
This way, we achieve both low memory and $o(n^{1+1/k})$ query time, as stated in the following theorem.

We remark again that Bilò et al.~\cite[Theorem~1.1]{Bilo23SubquadraticSpaceSTOC}
gave an oracle for unweighted graphs whose query time is sublinear,
at the expense of the space being only marginally subquadratic.

\begin{restatable}{theorem}{largediameter}
\label{thm:large_diameter}
	Let $G = (V,E)$ be an undirected graph 
	polynomial edge weights.
	For all positive integers $k$ and $f = o(\log n/\log\log n)$,
	and every  $0 < \alpha < 1$,
	there is an $f$-DSO for $G$ with stretch $2k-1$, space $kn^{1+\alpha + \frac{1}{k} +o(1)}$,
	preprocessing time $kmn^{1 + \alpha + \frac{1}{k} +o(1)}$,
	and query time $\Otilde(n^{1 + \frac{1}{k} -\frac{\alpha}{k(f+1)}})$.
\end{restatable}

\section{Preliminaries}
\label{sec:prelims}

\textbf{Shortest paths and hop diameter.}
We let $G = (V,E)$ denote the undirected base graph with $n$ vertices and $m$ edges,
edge-weighted by a function $w \colon E \to \mathcal{W}$,
where the set of admissible weights $\mathcal{W} \subseteq \mathbb{R}^+$ is of size $|\mathcal{W}| = \poly(n)$.
We tacitly assume $m = \OHmega(n)$.
For any undirected graph $H$ (that may differ from the input $G$)
we denote by $V(H)$ and $E(H)$ the set of its vertices and edges, respectively.
Let $P$ be a path in $H$ from a vertex $s \in V(H)$ to $t \in V(H)$,
we say that $P$ is an \emph{$s$-$t$-path} in $H$. 
We denote by $|P| = \sum_{e \in E(P)} w(e)$ the \emph{length} of $P$,
that is, its total weight.
For vertices $u,v \in V(P)$, we let $P[u..v]$ denote the subpath of $P$ from $u$ to $v$.
For two paths $P,Q$ in $H$ that share an endpoint, we use $P \circ Q$ for their concatenation. 
For $s,t \in V(H)$, the \emph{distance} $d_H(s,t)$ 
is the minimum length of any $s$-$t$-path in $H$;
if $s$ and $t$ are disconnected, we set $d_H(s,t) =+ \infty$.
When talking about the base graph $G$, we drop the subscripts
if this does not create any ambiguities.
The \emph{hop diameter} of $H$ is the maximum number of edges of
any shortest path between pairs of vertices in $V(H)$.

\vspace*{.5em}
\noindent
\textbf{Spanners and distance sensitivity oracles.}
A \emph{spanner of stretch} $\sigma \ge 1$, or \emph{$\sigma$-spanner}, 
for $H$ is a subgraph $S \subseteq H$ such that
for any two vertices $s,t \in V(S) = V(H)$, the inequality
$d_H(s,t) \le d_S(s,t) \le \sigma \cdot d_H(s,t)$ holds.
For a set $F \subseteq E$ of edges,
let $G{-}F$ be the graph obtained from $G$ by removing all edges in $F$.
For any two $s,t \in V$, a \emph{replacement path} $P(s,t,F)$ 
is a shortest path from $s$ to $t$ in $G{-}F$. 
Its length $d(s,t,F) = d_{G-F}(s,t)$ is the \emph{replacement distance}.
For a positive integer $f$, an \mbox{\emph{$f$-distance sensitivity oracle}} (DSO) reports,
upon query $(s,t,F)$ with $|F| \le f$, the replacement distance $d(s,t,F)$.
It has \emph{stretch} $\sigma \ge 1$, or is $\sigma$\emph{-approximate},
if the reported value $\widehat{d}(s,t,F)$
satisfies $d(s,t,F) \le \widehat{d}(s,t,F) \le \sigma \cdot d(s,t,F)$
for any admissible query.
We measure the space complexity of a data structure in the number of $O(\log n)$-bit machine words.
The size of the input $G$ does not count against the space 
unless it is stored explicitly.

\vspace*{.5em}
\noindent
\textbf{Error-correcting codes.}
For a positive integer $h$, we set $[h] = \{0,1, \dots, h-1\}$.
For positive integers $q$, $p$, and $\ell$ with $p \le \ell$, 
a \emph{code with alphabet size $q$, message length $p$, and block length $\ell$}
is a set $C \subseteq [q]^\ell$ such that $|C| \ge q^p$.
An \emph{encoding} for $C$ is a computable injective mapping $[q]^p \to C$.
Two codewords $x,y \in C$ have \emph{(relative) distance}
$\Delta(x,y) = |\{j \in [\ell] \mid x_j \neq y_j\}|/\ell$.
For a positive real $\delta > 0$, code $C$ is
\emph{error-correcting with (relative) distance} $\delta$,
if for any two $x,y \in C$, $\Delta(x,y) \ge \delta$.
In this case, we say $C$ is a $[p,\ell,\delta]_q$-code.
It will be sufficient to focus on Reed-Solomon codes,
which are $[p,q,1- \frac{p-1}{q}]_q$-codes for any $p \le q$.
When choosing $q$ (and therefore $\ell = q$) as a power of $2$ and $p < q$,
there is an encoding algorithm for Reed-Solomon codes
that takes $O(\ell \log p)$ time and $O(\ell)$ space
using fast Fourier transform~\cite{Lin16NovelPolynomialBasis}.

\section{Small Hop Diameter}
\label{sec:small_diameter}

We first describe the simpler randomized version of our distance sensitivity oracle
for graphs with small hop diameter.
Afterwards, we derandomize it using more involved techniques like error-correcting codes.
Throughout, we assume that the base graph $G$ has edge weights from a polynomial-sized range.

\subsection{Preprocessing}
\label{subsec:small_diameter_preprocessing}

In the setting of \Cref{thm:small_diameter},
all shortest paths have at most $D$ edges.
Let $f = o(\log n/\log\log n)$ be the sensitivity of the oracle and
$L \ge \max(f,2)$ be an integer parameter which will be fixed later (depending on $D$).
An $(L,f)$\emph{-replacement path covering} (RPC)~\cite{KarthikParter21DeterministicRPC}
is a family $\G$ of spanning subgraphs of $G$ such that
for any set $F \subseteq E$, $|F| \le f$, and pair of vertices $s,t \in V$ 
such that $s$ and $t$ have a shortest path in $G-F$ on at most $L$ edges,
there exists a subgraph $G_i \in \G$ that does not contain any edge of $F$
but an $s$-$t$-path of length $d(s,t,F)$.
That means, some replacement path $P(s,t,F)$ from $G-F$ also exists in $G_i$.
Let $\G_F \subseteq \G$ be the subfamily of all graphs that do not contain an edge of $F$.
The definition of an RPC implies that if $s$ and $t$ have a replacement path w.r.t.\ $F$ on at most $L$ edges,
then $\min_{G_i \in \G_F} d_{G_i}(s,t) = d(s,t,F)$.

To build the DSO, we first construct an $(L,f)$-RPC.
This can be done by generating $|\G| = cf L^f \ln n$
random subgraphs for a sufficiently large constant $c > 0$.
Each graph $G_i$ is obtained from $G$ by deleting any edge with probability $1/L$
independently of all other choices.
As shown in~\cite{WY13}, the family $\G = \{G_i\}_{i}$ is an $(L,f)$-RPC with high probability
It is also easy to see using Chernoff bounds\footnote{%
	There is a slight omission in \cite[Lemma~3.1]{WY13}
	for $|\G_F|$ is only calculated for $|F| = f$.
}
that for any failure set $F$, $|\G_F| = O(|\G|/L^{|F|}) = \Otilde(f L^{f-|F|})$.

We do not allow ourselves the space to store all subgraphs.
We therefore replace each $G_{i}$ by a distance oracle $D_{i}$,
a data structure that reports, for any two $s,t \in V$,
(an approximation of) the distance $d_{G_{i}}(s,t)$.
For any positive integer $k$, Thorup and Zwick~\cite{ThorupZ05} devised a DO
that is computable in time $\Otilde(kmn^{1/k})$, has size $O(kn^{1+1/k})$, query time $O(k)$,
and a stretch of $2k-1$.
Roddity, Thorup, and Zwick~\cite{RodittyThorupZwick05DeterministicDO} derandomized the oracle,
and Chechik~\cite{Chechik14,Chechik15} reduced the query time to $O(1)$
and the space to $O(n^{1+1/k})$.
Additionally, we store, for each $G_i$, a spanner $S_i$.
The same work~\cite{ThorupZ05} contains a spanner construction with stretch $2k-1$
that is compatible with the oracle,
meaning that the oracle $D_{i}$ reports exactly the value $d_{S_i}(s,t)$.
The spanner is computable in time $\Otilde(kmn^{1/k})$ and has $O(kn^{1+1/k})$ edges.
We store it as a set of edges.
There are static dictionary data structures known that achieve this in $O(kn^{1+1/k})$ space
such that we can check in $O(1)$ worst-case time whether an edge is present or retrieve an edge.
They can be constructed in time $\Otilde(kn^{1+1/k})$~\cite{HagerupMiltersenPagh01DeterministicDictionaries}.
The total preprocessing time of the distance sensitivity oracle
is $\Otilde(|\G|m + |\G|kmn^{1/k}) = \Otilde(f L^f k \nwspace mn^{1/k})$
and it takes $\Otilde(fL^f \nwspace kn^{1+1/k})$ space.

\subsection{Query Algorithm}
\label{subsec:small_diameter_query}

Assume for now that the only allowed queries to the DSO are triples $(s,t,F)$ 
of vertices $s,t \in V$ and a set $F \subseteq E$ of at most $f$ edges
such that any shortest path from $s$ to $t$ in $G-F$ has at most $L$ edges.
We will justify this assumption later with the right choice of $L$.
The oracle has to report the replacement distance $d(s,t,F)$.
Recall that $\G_F$ is the family of all graphs in $\G$ that have at least all edges of $F$ removed.
Since $\G$ is an $(L,f)$-RPC,
all we have to do is compute (a superset of) $\G_F$ and retrieve (an approximation of)
$\min_{G_i \in \G_F} d_{G_i}(s,t)$.
The issue is that we do not have access to the graphs $G_i$ directly.

The idea is to use the spanners $S_i$ as proxies.
This is justified by the next lemma
that follows from a connection between the spanners and oracles presented in~\cite{ThorupZ05}.
Let $D_i(s,t)$ denote the answer of the distance oracles $D_i$.

\begin{restatable}{lemma}{spannerproxy}
\label{lem:spanner_proxy}
	Let $G_i \in \G$ be a subgraph, $S_i$ its associated spanner, 
	and $D_i$ its $(2k{-}1)$-approximate distance oracle.
	For any two vertices $s,t \in V$ and set $F \subseteq E$ with $|F| \le f$,
	if $F \cap E(S_i) = \emptyset$, then
	$d(s,t,F) \le D_i(s,t) \le (2k{-}1) \, d_{G_i}(s,t)$.
\end{restatable}


%

Let $\G^S = \{S_i\}_{i \in [r]}$ be the collection of spanners for all $G_i \in \G$,
and $\G^S_F \subseteq \G^S$ those that do not contain an edge of $F$.
Below, we hardly distinguish between a set of spanners (or subgraphs) and their indices,
thus e.g.\ $S_i \in \G^S_F$ is abbreviated as $i \in \G^S_F$.
Since $E(S_i) \subseteq E(G_i)$ and using the convention, we get $\G_F^S \supseteq \G_F$.\footnote{%
	We do mean here that $\G_F^S$ is a superset of $\G_F$.
	Since the spanner contain fewer edges than the graphs, $F$ may be missing from $E(S_i)$
	even though $F \cap E(G_i) \neq \emptyset$.
	This is fine as long as we take the \emph{minimum} distance over all spanners from $\G_F^S$
}
To compute $\G_F^S$, we cycle through all of $\G^S$ and probe each dictionary with the edges in $F$,
this takes $O(f|\G|) = \Otilde(f^2 L^f)$ time and dominates the query time.
If $i \in \G^S_F$, then we query the distance oracle $D_i$ with the pair $(s,t)$ in constant time.
As answer to the query $(s,t,F)$, we return $\min_{i \in \G_F^S} D_i(s,t)$.
By  \Cref{lem:spanner_proxy}, the answer is at least as large as the sought replacement distance
and, since there is a graph $G_i \in \G_F \subseteq \G_F^S$ with $d_{G_i}(s,t) = d(s,t,F)$,
it is at most $(2k {-}1) \, d(s,t,F)$.

Let $D$ be an upper bound on the hop diameter of $G$.
As mentioned above, Afek et al.~\cite[Theorem~2]{Afek02RestorationbyPathConcatenation_journal}
showed that the maximum hop diameter of all graphs $G-F$ for $|F| \le f$ is bounded by $(f{+}1)D+f$.
Using this value for $L$ implies that indeed all queries admit a replacement path
on at most $L$ edges.
For the DSO, it implies a preprocessing time of 
$\Otilde(f L^f k \nwspace mn^{1/k}) = \Otilde(f^{f+1} D^f \nwspace k m n^{1/k})$,
which for $f = o(\log n/\log\log n)$ is $D^f kmn^{1/k+o(1)}$.
The space requirement is
$\Otilde(fL^f \nwspace kn^{1+1/k}) = \Otilde(f^{f+1} D^f \nwspace kn^{1+1/k}) = D^f kn^{1+1/k+o(1)}$,
and the query time $\Otilde(f^2 L^f) = \Otilde(f^{f+2} D^f) = D^f n^{o(1)}$.
This proves the first part of \Cref{thm:small_diameter}.

\subsection{Derandomization}
\label{subsec:small_diameter_derand}

We now make the DSO deterministic via a technique 
by Karthik and Parter~\cite{KarthikParter21DeterministicRPC}.
The derandomization will allow us to find the relevant subgraphs faster,
so we do not need the spanners anymore.
Recall that the distance oracles $D_i$ were already derandomized in~\cite{RodittyThorupZwick05DeterministicDO}. 
The only randomness left is the generation of the subgraphs $G_i$.
Getting a deterministic construction offers an alternative approach to dealing with the issue that 
discarding the subgraphs for space reasons deprives us of the information 
which edges have been removed.
Intuitively, we can now reiterate this process at query time to find the family $\G_F$.
Below we implement this idea in a space-efficient manner.

We identify the edge set $E = \{e_0,e_1, \dots, e_{m-1}\}$ with $[m]$.
Let $q$ be a positive integer.
Assume that $p = \log_q m$ is integral, otherwise one can replace $\log_q m$
with $\lceil \log_q m \rceil$ without any changes. 
We interpret any edge $e_i \in E$ as a base-$q$ number $(c_0,c_1, \dots, c_{p-1}) \in [q]^p$
by requiring $i = \sum_{j=0}^{p-1} c_j q^j$.
Consider an error-correcting $[p, \ell, \delta]$-code with distance 
$\delta > 1 - \frac{1}{fL}$
and (slightly abusing notation) let $C$ be the $(m {\times} \ell)$-matrix with entries in $[q]$
whose $i$-th row is the codeword encoding the message $e_i = (c_0,c_1,\dots,c_{p-1})$.
The key contribution of the work by Karthik and Parter~\cite{KarthikParter21DeterministicRPC}
is the observation that the \emph{columns} of $C$ form a family of hash functions 
$\{h_j \colon E \to [q]\}_{j \in [\ell]}$
such that for any pair of disjoint sets $P,F \subseteq E$ with $|P| \le L$ and $|F| \le f$
there exists an index $j \in [\ell]$ with
$\forall x \in P, y \in F \colon h_j(x) \neq h_j(y)$.

An $(L,f)$-replacement path covering can be constructed from this as follows.
Choose $q$ as the smallest power of $2$ greater\footnote{
	The original construction in~\cite{KarthikParter21DeterministicRPC}
	sets $q$ as a prime number.
	We use a power of $2$ instead to utilize the encoding algorithm in~\cite{Lin16NovelPolynomialBasis}.
	All statements hold verbatim for both cases.
} than $fL \log_L m$.
Note that $q \le \frac{2fL \log_2 m}{\log_2 L} \le \frac{4fL \log_2 n}{\log_2 L}$.
A Reed-Solomon code with alphabet size $q$, message length $p = \log_q m$, and block length $\ell = q$
has distance greater than $1-\frac{1}{fL}$~\cite[Corollary~18]{KarthikParter21DeterministicRPC}.
The resulting covering $\G$ consists of $|\G| = O(\ell \cdot q^f) = O(q^{f+1})$ subgraphs,
each one indexed by a pair $(j,S)$ where $j \in [\ell]$ and $S \subseteq [q]$ is a set with $|S| \le f$.
In the subgraph $G_{(j,S)}$, an edge $e_i$ is removed if and only if $h_j(e_i) \in S$.
It is verified in~\cite{KarthikParter21DeterministicRPC}
that the family $\G = \{G_{(j,S)}\}_{j,S}$ is indeed an $(L,f)$-RPC.
Moreover, for a fixed set $F = \{e_{i_1}, \dots, e_{i_{|F|}}\}$ of edge failures,
define $\G_F \subseteq \G$ to be the subfamily consisting
of the graphs indexed by $(j, \nwspace \{h_j(e_{i_1}), \dots, h_j(e_{i_{|F|}})\})$ 
for each $j \in [\ell]$.
Then, the construction ensures that no graph in $\G_F$ contains any edge of $F$
and, for each pair of vertices $s,t \in V$ with an replacement path (w.r.t.\ $F$)
on at most $L$ edges,
there is graph in $\G_F$ in which $s$ and $t$ are joined by a path of length $d(s,t,F)$.
The $\ell$ graphs in $\G_F$ contain all the information
we need for the short replacement distances with respect to the failure set $F$.
The number of subgraphs in the covering is $O((4fL \log_L n)^{f+1})$,
this is a factor $O((4f \log_L n)^f L)$ larger than what we had for the randomized variant.

In turn, we can make use of the extreme locality of the indexing scheme for $\G$.
Since we chose $q$ as a power of $2$, we get the letters of the message 
$e_i = (c_0,c_1,\dots,c_{p-1})$ by reading off blocks of $\log_2 q$ bits
of the binary representation of $i$.
The codeword of $C$ corresponding to $e_i$ is
computable in time $O(p + \ell \log p) = O(f L (\log_L n) \log\log n)$ 
and space $O(p+\ell) = O(f L \log_L n)$
with the encoding algorithm of Lin,  Al-Naffouri, Han, and Chung~\cite{Lin16NovelPolynomialBasis}.
The whole matrix $C$ and from it the family $\G$ can be generated in time
$O(f L m (\log_L n) \log\log n + |\G| \nwspace m) = O((4fL \log_L n)^{f+1} \nwspace m)$.
Note that the codeword of $e_i$ is $(h_1(e_i), h_2(e_i), \dots, h_\ell(e_i))$.
So even after discarding $C$ and the subgraphs, we can find the indices of graphs in $\G_F$
by encoding the edges of $F$ in time $O(|F| \ell \log p) = O(f^2 L (\log_L n) \log\log n)$
and rearranging the values into the $\ell$ sets $\{h_j(e_{i_1}), \dots, h_j(e_{i_{|F|}})\}$.
In particular, using the algorithm in \cite{Lin16NovelPolynomialBasis},
we do not need to store the generator matrix of the code $C$.

The remaining preprocessing is similar as in \Cref{subsec:small_diameter_preprocessing},
but we need neither the spanners nor the dictionaries anymore.
We set $L = O(fD)$ again and,
for each subgraph $G_{(j,S)}$, we only build the distance oracle $D_{(j,S)}$.
This dominates the preprocessing time
$\Otilde(|\G| kmn^{1/k}) = \Otilde( 4^{f+1} f^{2f+2} D^{f+1} (\log_D n)^{f+1} kmn^{1/k})\\
	= D^{f+1} kmn^{1/k+o(1)}$.
The total size is now $\Otilde(|\G| kn^{1+1/k}) = D^{f+1} kn^{1+1/k+o(1)}$.
However, due to the derandomization, the query time is now much faster than before,
in particular, polynomial in $f$, $D$ and $\log n$.
We do not have to cycle through all spanners anymore
and instead compute $\G_F$ and query the DO only for those $\ell$ graphs.
As a result, the time to report the replacement distance is
$O(f^2 L (\log_L n) \log\log n + \ell) = O(f^3 D \nwspace (\log_D n) \log\log n)$,
completing \Cref{thm:small_diameter}.

\section{Large Hop Diameter}
\label{sec:large_diameter}

We also devise a distance sensitivity oracle for graphs with an arbitrary hop diameter
while maintaining a small memory footprint.
For this, we have to handle hop-long replacement paths, that is, those that have more than $L$ edges.
We obtain a subquadratic-space distance sensitivity oracle with the same stretch of $2k-1$ but an $o(n^{1+1/k})$ query time.
This is faster than computing the distance in any possible spanner.

\subsection{Deterministic Pivot Selection}
\label{subsec:large_diameter_pivots}

We say a query with vertices $s,t \in V$ and set $F \subseteq E$, $|F| \le f$
has long replacement paths if every $P(s,t,F)$ has at least $L$ edges.
Those need to be handled in general DSOs.
This is usually done by drawing a random subset $B \subset V$ of 
$\Otilde(fn/L)$ pivots, as in~\cite{WY13}, or essentially equivalent
sampling every vertex independently with probability $\Otilde(f/L)$~\cite{GrandoniVWilliamsFasterRPandDSO_journal,RodittyZwick12kSimpleShortestPaths}.
With high probability, $B$ hits one path for every query with long replacement paths.

There are different approaches known to derandomize this depending on the setting~\cite{AlonChechikCohen19CombinatorialRP,BeKa09,Bilo22Extremal,BCFS21SingleSourceDSO_ESA,King99FullyDynamicAPSP}.
In our case, we can simply resort to the replacement path covering to obtain $\P$
since we have to preprocess it anyway.
We prove the following lemma for the more general class of arbitrary positive weights.
Note that the key properties of an $(L,f)$-RPC remain in place as all definitions
are with respect to the number of edges on the replacement paths.
We make it so that $B$ hits the slightly shorter paths with $L/2$ edges (instead of $L$).
We are going to use this stronger requirement in \Cref{lem:spanner_for_pivots}.

\begin{restatable}{lemma}{derandpivots}
\label{lem:derand_pivots}
	Let $G = (V,E)$ be an undirected graph with positive edge weights.
	Let $Q$ be the set of all queries $(s,t,F)$,
	with $s,t \in V$ and $F \subseteq E$, $|F| \le f$, 
	for which every $s$-$t$-replacement path w.r.t.\ $F$ has at least $L/2$ edges.
	Given an $(L,f)$-replacement path covering $\G$ for $G$,
	there is a deterministic algorithm that computes in time 
	$\Otilde(|\G| (mn + Ln^2/f))$ a set $B \subseteq V$ of size $\Otilde(fn/L)$ such that,
	for all $(s,t,F) \in Q$,
	there is a replacement path $P = P(s,t,F)$ with $B \cap V(P) \neq \emptyset$.\vspace*{.25em}
	At the same time, one can build a data structure of size
	$O(|\G||B|^2)$ that reports, for every $G_i \in \G$ and $x,y \in B$, 
	the distance $d_{G_i}(x,y)$ in constant time.
\end{restatable}

\subsection{Preprocessing}
\label{subsec:large_diameter_preprocessing}

Our solution for large hop diameter builds on the deterministic DSO in \Cref{subsec:small_diameter_derand}.
As for the case of a small hop diameter,
we construct an \mbox{$(L,f)$-}replace\-ment path covering $\G$ and,
for each $G_{(i,S)} \in \G$, the distance oracle $D_{(i,S)}$.
Recall that this part takes time $\Otilde(|\G|kmn^{1/k})$
and $O(|\G|kn^{1+1/k})$ space.

We invoke \Cref{lem:derand_pivots} to obtain the set $B$.
Additionally, for each subgraph, we build a complete weighted graph $H_{(i,S)}$ on the vertex set $B$
where the weight of edge $\{x,y\}$ is $d_{G_{(i,S)}}(x,y)$,
which we retrieve from the data structure mentioned in  \Cref{lem:derand_pivots}.
We then compute a $(2k{-}1)$-spanner $T_{(i,S)}$ for $H_{(i,S)}$ with  $O(k|B|^{1+1/k})$ edges
via the same deterministic algorithm by Roditty, Thorup, and Zwick~\cite{RodittyThorupZwick05DeterministicDO}.
We store the new spanners for our DSO.
The time to compute them is $\Otilde(|\G|k|B|^{2+1/k})$ and,
since $|B| = \Otilde(fn/L)$, the preprocessing time is
\begin{gather*}
	\Otilde(|\G| (kmn^{1/k}+ |B|m + k|B|^{2+1/k}))
		= \Otilde(|\G| (mn + k|B|^{2+1/k})) \le\\
		L^{f+1} mn^{1+o(1)} + L^{f-1-1/k} k n^{2+1/k+o(1)}
		\le L^{f+1} kmn^{1 +1/k + o(1)}.
\end{gather*}

\noindent
To obtain the bounds of \Cref{thm:large_diameter}, we set $L = n^{\frac{\alpha}{f+1}}$.
Parameter $0 < \alpha <1$ allows us to balance the space and query time.
With this, we get a preprocessing time of $kmn^{1 + \alpha + 1/k + o(1)}$.
and a space of
$O(|\G|(kn^{1+1/k} + n + k|B|^{1+1/k}) = \Otilde(|\G| kn^{1+1/k}) 
	= L^{f+1}kn^{1+1/k+o(1)} = kn^{1+ \alpha + 1/k+o(1)}$.

\subsection{Updated Query Algorithm}
\label{subsec:large_diameter_query}

The algorithm to answer a query $(s,t,F)$ starts similarly as before.
We use the error-correcting codes to compute the subfamily $\G_F$
and the estimate $\widehat{d_1}(s,t,F) = \min_{(i,S) \in \G_F} D_{(i,S)}(s,t)$ is retrieved.
However, this is no longer guaranteed to be an $(2k{-}1)$-approximation
if the query is hop-long, i.e., if every shortest $s$-$t$-path in $G-F$ has at least $L$ edges.
It could be that no replacement paths survive and, in the extreme case,
$s$ and $t$ are disconnected in each $G_{(i,S)} \in \G_F$,
while they still have a finite distance in $G-F$.
To account for long queries, we join all the spanners $T_i$ for $i \in \G_F$.
In more detail, we build a multigraph\footnote{%
	The multigraph is only used to ease notation.
}
$H^F$ on the vertex set $V(H^F) = B \cup \{s,t\}$
whose edge set (restricted to pairs of pivots) is the disjoint union of all the sets $\{E(T_i)\}_{i \in \G_F}$
and, for each subgraph $(i,S) \in \G_F$ and pivot $x \in B$ contains the edges $\{s, x\}$ and $\{x,t\}$
with respective weights $D_{(i,S)}(s, x)$ and $D_{(i,S)}(x,t)$, where $D_{(i,S)}$ is the corresponding DO.
The oracle then computes the second estimate $\widehat{d_2}(s,t,F) = d_{H^F}(s,t)$
and returns $\widehat{d}(s,t,F) = \min \{\widehat{d_1}(s,t,F), \; \widehat{d_2}(s,t,F)\}$.

\begin{restatable}{lemma}{querytimelargediameter}
\label{lem:query_time_large_diameter}
	The distance sensitivity oracle has stretch $2k-1$ and the query takes
	time $\Otilde(\frac{n^{1+1/k}}{L^{1/k}}) = \Otilde(n^{1+\frac{1}{k} - \frac{\alpha}{k(f+1)}})$.
\end{restatable}

Proving this lemma is enough to complete \Cref{thm:large_diameter}.
In order to do so, we first establish the fact that $d_{H^F}(s,t)$
is a $(2k{-}1)$-approximation for $d(s,t,F)$
in case of a long query.

\begin{restatable}{lemma}{spannerforpivots}
\label{lem:spanner_for_pivots}
	Let $s,t \in V$ be two vertices and $F \subseteq E$ a set of edges with $|F| \le f$.
	It holds that $d(s,t,F) \le d_{H^F}(s,t)$.
	If additionally every shortest $s$-$t$-path in $G-F$ has more than $L$ edges,
	then we have $d_{H^F}(s,t) \le (2k{-}1) \, d(s,t,F)$.
\end{restatable}

\subsubsection{Acknowledgements.} 

The authors thank Merav Parter for raising the question of designing distance sensitivity oracles that require only subquadratic space.
\vspace*{.75em}

\begin{minipage}{0.96\textwidth}
\begin{wrapfigure}{l}{3.2cm}
\flushleft
\vspace*{-2.25em}
\hspace*{-1.5em}
	\includegraphics[scale=.125]{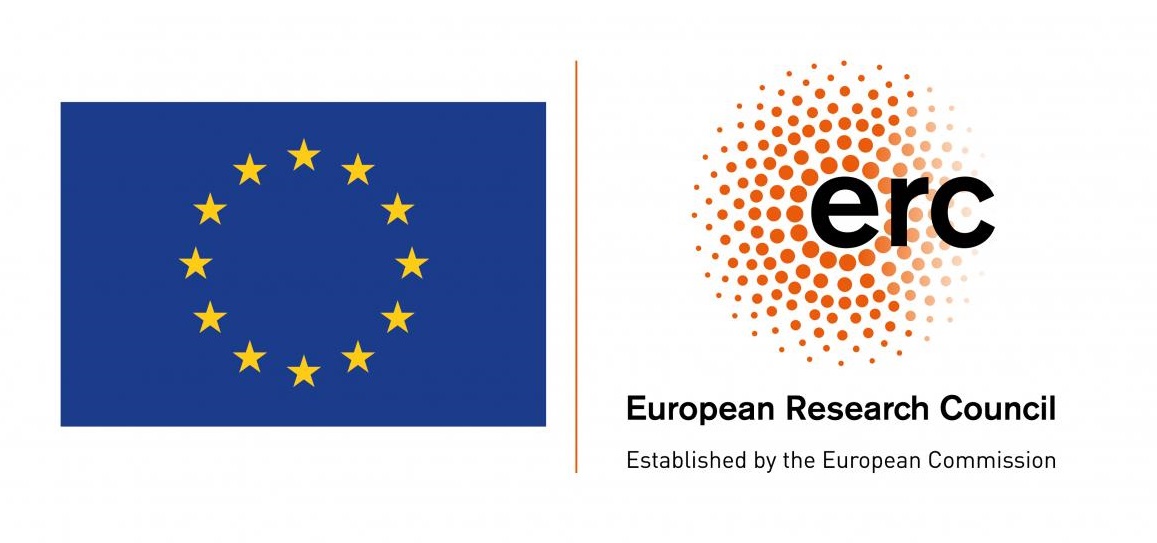}
\end{wrapfigure}

This project received funding from the
European Research Council (ERC) under the European Union's Horizon
2020 research and innovation program (Grant agreement No.~101019564
``The Design of Modern Fully Dynamic Data Structures (MoDynStruct)'').
\end{minipage}

\bibliographystyle{splncs04}
\bibliography{DSO_bib}

\newpage
\appendix

\section{Omitted Proofs}
\label{app:omitted_proofs}

\subsection{Proofs Omitted from \Cref{sec:small_diameter}}

\spannerproxy*

\begin{proof}
	The construction of $S_i$ and $D_i$ is such that $d_{S_i}(s,t) = D_i(s,t)$
	is satisfied for all $s,t \in V$~\cite[Corollary 4.3]{ThorupZ05}.
	The work~\cite{RodittyThorupZwick05DeterministicDO} merely makes the spanner/oracle
	construction deterministic,
	and the query time/size reduction of the oracle in~\cite{Chechik14,Chechik15} only changes
	how the distances are queried/stored, but not the distances themselves.
	Together with the assumption $F \cap E(S_i) = \emptyset$, whence $S_i \subseteq G{-}F$, 
	we get $d(s,t,F) = d_{G-F}(s,t) \le d_{S_i}(s,t) \le D_i(s,t)$.
	
	The other inequality is immediate from $D_i$ being a $(2k{-}1)$-approximate DO.
\end{proof}

\subsection{Proofs Omitted from \Cref{sec:large_diameter}}

\derandpivots*

\begin{proof}
	The algorithm first computes all-pairs shortest paths in all graphs of the covering $\G$
	in such a way that the distances and also the paths themselves can be reported in constant time (per edge)
	and we additionally get access to the number of edges on the paths.
	This can be done by computing a shortest path tree rooted at each vertex in each graph using Dijkstra's algorithm,
	with some simple bookkeeping for the number of edges,
	in time $\Otilde(|\G| mn)$.

	We then build a set $\P$ of paths.
	For each graph $G_i \in \G$ and vertices $u,v \in V$ such that the shortest $u$-$v$-path in $G_i$ 
	computed above has at least $K = L/(4f{+}2)$ edges,
	we add an (arbitrary) $K$-edge subpath to $\P$.
	This takes additional time $O(K|\G|n^2) = O(|\G|Ln^2/f)$.
	Eventually, $\P$ contains $O(|\G|n^2)$ paths each of size $K$.
	The folklore greedy algorithm, i.e.\ always choosing the vertex from $V$ that hits the most unhit paths,
	gives a hitting set $B$ for $\P$ with $|B| = \Otilde(n/K) = \Otilde(fn/L)$ elements
	in time $\Otilde(K|\P|) = \Otilde(|\G|Ln^2/f)$, see e.g.~\cite{AlonChechikCohen19CombinatorialRP,King99FullyDynamicAPSP}.
	The data structure is merely an array indexed by $\G \times \binom{B}{2}$ holding the distances $d_{G_i}(x,y)$.
	
	It is left to prove that the pivots in $B$ hit all the relevant paths.
	Let $s,t \in V$, and $F \subseteq E$ with $|F| \le f$ be such 
	that every shortest path from $s$ to $t$ in $G-F$ has at least $L/2$ edges.
	Recall that in \emph{unweighted} (undirected) graphs, any replacement path decomposes into 
	at most $f+1$ shortest paths in the original graph $G$.
	Afek et al.~\cite[Theorem~2]{Afek02RestorationbyPathConcatenation_journal}
	also gave an analog for \emph{weighted} (undirected) graphs.
	There, one can find a decomposition into $f+1$ shortest paths in $G$
	interleaved by up to $f$ edges.
	For a replacement path on at least $L/2$ edges,
	this means one of its subpaths must have at least $K = \frac{L/2}{(f+1)+f}$ edges.
	Let $P' = P'(s,t,F)$ be a replacement path with the minimum number of edges.
	By the preceding argument, it contains two vertices $u,v \in V(P')$ such that the segment
	$P'[u..v]$ is a shortest path in $G$ with exactly $K$ edges
	and \emph{every} shortest $u$-$v$-path has at least $K$ edges.

	Let $G_i \in \G$ be a subgraph that contains no edge of $F$ and $d_{G_i}(u,v) = d(u,v,F)$,
	such a graph exists because $\G$ is an $(L,f)$-RPC. 
	Let $Q$ be the shortest $u$-$v$-path computed for $G_i$ in the first step.
	It may differ from $P'[u..v]$, but is as the same length $d(u,v,F) = d(u,v)$
	and at least $K$ edges.
	We therefore added (a subpath of) $Q$ to $\P$ and there is a pivot $x \in B \cap V(Q)$.
	The replacement path $P = P'[s..u] \circ Q \circ P'[v..t]$ satisfies the assertion of the lemma.
\end{proof}

The proof of \Cref{lem:query_time_large_diameter} relies on \Cref{lem:spanner_for_pivots} (see below).
We decided to present the proofs in the order in which the lemmas appear in the main part.

\querytimelargediameter*

\begin{proof}
	Let $s,t \in V$ be two vertices and $F \subseteq E$ a set of at most $f$ failures.
	Recall that the DSO computes and answers $\widehat{d} = \min\{\widehat{d_1},\widehat{d_2}\}$,
	where $\widehat{d_1}(s,t,F) = \min_{(i,S) \in \G_F} D_{(i,S)}(s,t)$ and $\widehat{d_2}(s,t,F) = d_{H^F}(s,t)$.
	Both estimates are at least as large as $d(s,t,F)$,
	the first by \Cref{lem:spanner_proxy} and the second one by the first clause of \Cref{lem:spanner_for_pivots}.
	Since $\G$ is an $(L,f)$-replacement path covering, 
	if there is a shortest path from $s$ to $t$ in $G-F$ with at most $L$ edges, then
	$\widehat{d}(s,t,F) \le  \min_{(i,S) \in \G_F} D_{(i,S)}(s,t)
		\le (2k-1) \cdot \min_{(i,S) \in \G_F} d_{G_{(i,S)}}(s,t)
		= (2k-1) \cdot d(s,t,F)$.
	Where the middle inequality follows from $D_{(i,S)}$ being a DO for $G_{(i,S)}$ with stretch $2k-1$.
	However, if all replacement paths $P(s,t,F)$ have more than $L$ edges,
	if could be that none of them survives in any $G_{(i,S)}$, whence $\widehat{d_1}(s,t,F)$ may overestimate $d(s,t,F)$
	by more than a factor $2k-1$.
	In this case, \Cref{lem:spanner_for_pivots} shows that $\widehat{d}(s,t,F) \le \widehat{d_2}(s,t,F) \le (2k{-}1) \, d(s,t,F)$.
	
	We are left to prove the query time.
	We have seen in \Cref{subsec:small_diameter_derand} that the subfamily $\G_F$
	contains $O(fL \nwspace \frac{\log n}{\log L})$ subgraphs
	and can be computed in time $O(f^2 L \nwspace \frac{\log n \log\log n}{\log L})$.
	Within the same bounds, we get the estimate $\widehat{d_1}(s,t,F)$
	as the distance oracles $D_{(i,S)}$ answer in constant time.
	It takes much more time to compute the second estimate.
 	We merge the $O(k |B|^{1+1/k}) = \Otilde(f^{1+1/k} k \nwspace \frac{n^{1+1/k}}{L^{1+1/k}})$ edges of the spanner $T_{(i,S)}$
 	into $H^F$ for each graph in $\G_F$.
	Afterwards, we add $O(|\G_F||B|) = \Otilde(fL \nwspace \frac{\log n}{\log L}  \cdot f\frac{n}{L}) = \Otilde(f^2 n)$ 
	more edges of the types $\{s,x\}$ or $\{x,t\}$ for $x \in B$.
	Eventually, $H^F$ has $\Otilde(f^{2+1/k} k \, \frac{n^{1+1/k}}{L^{1/k}})$ edges,
	and the time to compute $d_{H^F}(s,t)$ differs from that only by logarithmic factors.
	By the assumption $f = o(\log n/\log\log n)$,
	this is $\Otilde(n^{1+1/k}/L^{1/k})$.
\end{proof}

\spannerforpivots*

\begin{proof}
	The first inequality is an easy consequence of the fact that 
	none of the subgraphs in $\G_F$ contains an edge of $F$ by construction, see~\cite{KarthikParter21DeterministicRPC}.
	
	Let $(s,t,F)$ be such that every $s$-$t$-replacement path has more than $L$ edges.
	\autoref{lem:derand_pivots}	guarantees that there is an $s$-$t$-replacement path that contains a pivot from $B$.
	We will use it to show the stronger fact that there is a replacement path
	such that the subpath between consecutive pivots, the one from $s$ to the first pivot, and from the last to $t$,
	all have at most $L$ edges.
	For the sake of this proof, we refer to the elements of $B \cup \{s,t\}$ as the \emph{extended pivots}.
	Suppose $P = P(s,t,F)$ is a replacement path that does not satisfy the above property
	and let $x,y \in V(P)$ be consecutive extended pivots such that $|E(P[x..y])| > L$.
	We distinguish two cases:
	either there is a replacement path $P'(x,y,F)$ with less than $L/2$ edges,
	or all of them have at least $L/2$.
	
	In the first case, we can swap the segment $P[x..y]$ for $P'(x,y,F)$.
	The new path is still a $s$-$t$-replacement path due to both parts having the same weight,
	but the hop distance between $x$ and $y$ is halved.
	For the second case, let $u$ be the second vertex on $P[x..y]$ (the one after $x$) and $v$ the second-to-last.
	The subsegment $P[u..v]$ has at least $L-1 \ge L/2$ edges,
	hence, by \Cref{lem:derand_pivots}, there is a $P''(u,v,F)$ that contains (an actual) pivot $z \in B$.
	Observe that $z$ must be different from both $x$ and $y$
	as otherwise there would be a shorter replacement path between $x$ and $y$.
	Again, swapping $P''(u,v,F)$ for $P[u..v]$ maintains the replacement path property,
	but reduces the number of edges between consecutive extended pivots by at least $1$.
	Iterating this argument gives the desired path.
		
	We use $P$ to for the constructed $s$-$t$-replacement path and
	$x_1, x_2, \dots, x_{j_{\max}}$ its (actual) pivots in the order going from $s$ to $t$.
	Recall that $|E(P[x_{j-1}..x_{j}])| \le L$.
	Therefore, for each consecutive pair $x_{j-1},x_{j}$,
	there is some $(i,S) \in \G_F$ for which the edge $\{x_{j-1},x_{j}\}$ in the dense subgraph $H_{(i,S)}$ had the correct weight 
	$d_{G_{(i,S)}}(x_{j-1},x_{j}) = d(x_{j-1},x_{j},F)$.
	This was stretched by at most a factor $2k-1$ in the spanner $T_{(i,S)}$ 
	and the multigraph $H^F$ inherited this stretch.
	Applying this to all consecutive pairs gives 
	$d_{H^F}(x_1,x_{j_{\max}}) \le (2k{-}1) \, d(x_1,x_{j_{\max}},F)$.
	
	We still have to get from $s$ to the first pivot $x_1$
	(the argument for $x_{j_{\max}}$ and $t$ is the same).
	As shown above, the segment $P[s..x_1]$ also has at most $L$ edges.
	So there is an $(i,S) \in \G_F$ with $D_{(i,S)}(s,x_1) \le (2k-1) \, d_{G_{(i,S)}}(s,x_1) = (2k-1) \, d(s,x_1,F)$.
	The multigraph $H^F$ contains an edge $\{s,x_1\}$ with weight $D_{(i,S)}(s,x_1)$.
	Combining those facts shows that there is some path in $H^F$ from $s$ to $t$
	that has length at most
	\begin{gather*}
		d_{H^F}(s,x_1)+d_{H^F}(x_1,x_{j_{\max}-1})+d_{H^F}(x_{j_{\max}},t)\\
			\quad\quad\ \  \le (2k-1) \cdot \Big((d(s,x_1,F)+ d(x_1,x_{j_{\max}},F) + d(x_{j_{\max}},t,F) \Big)\\
			\quad\quad\ \  = (2k-1) \cdot |P| = (2k-1) \cdot d(s,t,F).
	\end{gather*}
\end{proof}

\end{document}